\documentclass[conference, twocolumn]{IEEEtran}
\usepackage{amsmath,amssymb,graphicx,epsfig,fancyhdr,amsthm,tabulary}
\usepackage{cite}
\usepackage{multirow}
\usepackage{float}
\usepackage{amsfonts}
\usepackage{footnote}
\makesavenoteenv{tabular}
\usepackage{tabulary,bigstrut,array,multirow,booktabs}%
\usepackage{tikz}
\usetikzlibrary{calc}
\usepackage{pdfsync}
\usepackage{bm}

\newtheorem{thm}{Theorem}[]
\newtheorem{cor}{Corollary}
\newtheorem{lem}{Lemma}
\newtheorem{pro}{Proposition}

\theoremstyle{remark}
\newtheorem{rem}{Remark}

\theoremstyle{definition}

\newcommand{\tr}{\mathop{\mathrm{tr}}}

\begin{document}

\title{Systematic DFT Frames: Principle and \\Eigenvalues Structure}

\author{\IEEEauthorblockN{Mojtaba Vaezi and Fabrice Labeau}\\
\IEEEauthorblockA{Department of Electrical and Computer Engineering\\
McGill University\\
Montreal, Quebec H3A 2A7, Canada\\
Email: mojtaba.vaezi@mail.mcgill.ca, fabrice.labeau@mcgill.ca}
}


\maketitle

\begin{abstract}
Motivated by a host of recent applications requiring some amount of redundancy,
{\it frames} are becoming a standard tool in the signal processing toolbox. In this paper,
we study a specific class of frames, known as discrete Fourier transform (DFT) codes, and introduce the notion of {\it systematic}
frames for this class. This is encouraged by application of systematic DFT codes
in distributed source coding using DFT codes, 
a new application for frames.
Studying their extreme eigenvalues, we show that, unlike DFT frames, systematic DFT frames are not necessarily {\it tight}. Then, we come up with conditions
for which these frames can be tight. 
In either case, the best and worst systematic frames are established from reconstruction error point of view.
Eigenvalues of DFT frames, and their subframes, play a pivotal role in this work.
\end{abstract}


\IEEEpeerreviewmaketitle

\section{Introduction}
{\let\thefootnote\relax\footnotetext{This work was supported by Hydro-Qu\'{e}bec, the Natural Sciences and Engineering Research Council of Canada and McGill University in the framework of the NSERC/Hydro-Qu\'{e}bec/McGill Industrial Research Chair in Interactive Information Infrastructure for the Power Grid.}}

A {\it basis} is a set of vectors that is used to ``uniquely" represent any vector as a linear combination of basis elements.
{\it Frames}, as opposed to bases, are ``redundant" set of vectors which are used for signal representation.
Therefore, frames are more general than bases as they are not necessarily linearly independent, but are complete.
What would be the benefit of representing a signal with more than the minimum number of vectors required for completeness?
Frames offer flexibility in design and have variety of applications.
They show resilience to additive noise (including quantization noise),
robustness to erasure (loss), and numerical stability of reconstruction \cite{kovacevic2007life}, \cite{goyal2001quantized}.
With increasing applications, frames are becoming more prevalent in signal processing.


 Frames are generally motivated by applications requiring
some level of redundancy. Among them is distributed source coding (DSC) that uses {\it DFT
codes}, recently introduced in \cite{Vaezi2011DSC}. This provides a new application for frame expansion, viewing
the generator matrix of a DFT code as a frame operator \cite{rath2004frame}.
In this paper, we consider this specific type of tight frames which are
known as DFT codes and used for erasure and error correction in the real field  \cite{goyal2001quantized, Vaezi2011DSC, rath2004frame, Marshall}.

Motivated by its application in parity-based DSC that uses DFT codes \cite{Vaezi2011DSC}, we introduce the notion of {\it systematic frames}.
A systematic frame is defined to be a frame that includes the identity matrix as a subframe.
Since {\it tight} frames minimize reconstruction error \cite{goyal2001quantized}, \cite{kovacevic2007life},
we explore {\it systematic tight DFT frames}.
Although it is straightforward to come up with systematic DFT frames, we show that systematic ``tight'' DFT frames exist only for specific DFT codes.
When such a frame does not exist, we will be looking for systematic DFT frames with the ``best''
performance, from minimum mean-squared reconstruction error standpoint. 
We also demonstrate which systematic frames are the ``worst'' in this sense.

Central to this paper is the properties of the {\it eigenvalues} of $V^HV$,
in which $V$ is a square or non-square submatrix of a DFT frame.\footnote{Note that eigenvalues of $V^H V$ and $VV^H$
are equal for a square $V$; also, $V^H V$ and $VV^H$ have the same nonzero eigenvalues for a non-square $V$.}
Specifically, we present some bounds on the extreme eigenvalues of such matrices. These bounds are used to determine the conditions
required for a systematic frame so as to be tight. Besides, eigenvalues are crucial in establishing the best and worst systematic frames.

The paper is organized as follows. In Section~\ref{sec:def}, we
review two set of inequalities on the eigenvalues of Hermitian matrices which are frequently used in this paper.
In Section~\ref{sec:DFTcodes}, we introduce systematic DFT frames and set
the ground to study their extreme eigenvalues in Section~\ref{sec:eig}.
Section~\ref{sec:class} is devoted to the evaluation of reconstruction error
and classification of systematic frames based on that. We conclude in Section~\ref{sec:con}.

\section{Definitions and preliminaries}
\label{sec:def}

An $n\times n$ Vandermonde matrix with unit complex entries is defined by
\begin{align}
W \triangleq \frac{1}{\sqrt{n}}\left( \begin{array}{ccccccc}
       1   &1 & \cdots & 1 \\
       e^{j\theta_1}   &e^{j\theta_2} & \cdots & e^{j\theta_n} \\
       \vdots & \vdots & \ddots & \vdots \\
       e^{j(n-1)\theta_1}   &e^{j(n-1)\theta_2}  & \cdots &e^{j(n-1)\theta_n}  \\
      \end{array}
\right)  \label{eq:Vand2}
\end{align}
in which $\theta_p \in [0, 2 \pi) $ and $\theta_p \neq \theta_q $ for $p\neq q$, $1\leq p,q \leq n$.
If $\theta_p = \frac{2\pi}{n}(p-1)$, $W$ becomes the well-known IDFT matrix \cite{mitra1998digital}.
For this Vandermonde matrix we can write \cite{tucci2011eigenvalue2}, \cite{tucci2011eigenvalue}
\begin{align}
 \det(WW^H) = |\det(W)|^2 =\frac{1}{n^n}\prod_{\substack 1\leq p<q\leq n}{|e^{i\theta_p}-e^{i\theta_q}|^2}
 \label{eq:detV0}
 \end{align}

Let $A$ be a Hermitian $k\times k$ matrix with real eigenvalues
$\{\lambda_{1}(A), \hdots, \lambda_{k}(A)\}$ which are collectively called the {\it spectrum} of $A$, and
assume $\lambda_{1}(A) \geq \lambda_{2}(A)\geq \cdots \geq \lambda_{k}(A)$.

Schur-Horn inequalities show to what extent the eigenvalues of
a Hermitian matrix constraint its diagonal entries.
\begin{pro}\textit{Schur-Horn inequalities} \cite{seber2008matrix} \\
Let $A$ be a Hermitian $k\times k$ matrix with real eigenvalues
$\lambda_{1}(A) \geq \lambda_{2}(A)\geq \cdots \geq \lambda_{k}(A).$
Then,  for any $1 \leq i_1 < i_2< \cdots < i_l \leq k$,
\begin{align}
 \lambda_{k-l+1}(A) + \cdots + \lambda_{k}(A) &\leq
 a_{i_1i_1}+ \cdots + a_{i_li_l} \nonumber \\
 &\leq  \lambda_{1}(A) + \cdots + \lambda_{l}(A),
 \label{eq:Schur-Horn}
 \end{align}
 where $ a_{11}, \hdots, a_{kk}$ are the diagonal entries of $A$.
Particularly, for $l=1$  and $l=k$ we obtain
 \begin{align}
 \label{eq:Schur-Horn1}
 \lambda_{k}(A) \leq  a_{11} \leq  \lambda_{1}(A), \\
  \sum_{\substack i=1}^{k}{\lambda_{i}(A)} = \sum_{\substack i=1}^{k}{a_{ii}}.
 \label{eq:Schur-Horn2}
 \end{align}
\end{pro}

 Another basic question in linear algebra asks the degree to which the eigenvalues of two Hermitian matrices constrain the eigenvalues of their sum.
Weyl's theorem gives an answer to this question in the following set of inequalities.

\begin{pro}{Weyl inequalities \cite{seber2008matrix}} \\
Let $A$ and $B$ be two Hermitian $k\times k$ matrices with spectrums $\{\lambda_1(A),\ldots,\lambda_k(A)\}$ and $\{\lambda_1(B),\ldots,\lambda_k(B)\}$, respectively.
Then, for $i,j\leq k$, we have
\begin{align}
\label{eq:Weyl1}
 \lambda_{i}(A+B) \leq  \lambda_{j}(A) + \lambda_{i-j+1}(B) \qquad \text{for} \quad j \leq i, \\
 \lambda_{i}(A+B) \geq  \lambda_{j}(A) + \lambda_{k+i-j}(B) \qquad \text{for} \quad j \geq i.
 \label{eq:Weyl2}
 \end{align}
 \end{pro}


\section{DFT Frames}
\label{sec:DFTcodes}
\subsection{BCH-DFT Codes}

DFT codes \cite{Marshall}, are linear block codes over the complex field
whose parity-check matrix $H$ is defined based on the DFT matrix.
A Bose-Chaudhuri-Hocquenghem (BCH) DFT code is a DFT code that insert $d-1$ cyclically adjacent
zeros in the spectrum of any codevector where $d$ is the designed
distance of that code \cite{Blahut2003algebraic}.
Real BCH-DFT codes, a subset of complex BCH-DFT codes, benefit from
a generator matrix with real entries.
The generator matrix of an $(n,k)$ real\footnote{In a real BCH-DFT code, $n$ and $k$ cannot be even simultaneously \cite{Marshall}.} BCH-DFT code is typically defined by
\begin{align}
G= \sqrt{\frac{n}{k}} W_n^H \Sigma W_k,
\label{eq:G1}
\end{align}

\noindent in which $W_l$ represents the DFT matrix of size $l$, and
$\Sigma$ is 
\begin{align}
\Sigma_{n\times k} = \left( \begin{array}{ccccccc}
       I_\alpha  & \bm{0}  \\
       \bm{0}   & \bm{0}  \\
      \bm{0}    &   I_\beta    \\
      \end{array}
\right), \label{eq:cov}
\end{align}
where $\alpha = \lceil n/2\rceil  -\lfloor (n-k)/2\rfloor$, $\beta=k-\alpha$,
and the sizes of zero blocks are such that $\Sigma$ is an $n\times k$ matrix \cite{gabay2007joint}.
One can check that $\Sigma^H \Sigma= I_k$, and $\Sigma \Sigma^H $ is an $n\times n $ matrix given by
\begin{align}
\Sigma\Sigma^H = \left( \begin{array}{ccccccc}
       I_\alpha  & \bm{0} & \bm{0}  \\
       \bm{0}   & \bm{0} & \bm{0} \\
      \bm{0}    & \bm{0} &   I_\beta    \\
      \end{array}
\right). \label{eq:cov1}
\end{align}
\label{Sigma}
Then, it is easy to see that,
\begin{align}
\label{eq:GHG}
G^HG &= \frac{n}{k} I_k, \\
GG^H&= \frac{n}{k} W_n^H \Sigma \Sigma^HW_n.
\label{eq:GGH}
\end{align}

One can view the generator matrix $G$ in \eqref{eq:G1} as an analysis frame operator.
The following lemma presents some properties of $GG^H$ which are
central for our results in the next section.

\begin{lem}
Let $ G_{p\times k}$ be a matrix consisting of $p$ arbitrary rows of $G$ defined by \eqref{eq:G1}.
Then, the following statements hold:\\
i. $GG^H$ is a Toeplitz and circulant matrix\\
ii. $ G_{p\times k}G_{p\times k}^H, 1< p < n$ is a Toeplitz matrix\\
iii. All principal diagonal entries of $ G_{p\times k}G_{p\times k}^H, 1\leq p \leq n$ are equal to $1$.
\label{lem1}
\end{lem}

\begin{proof}
Let $a_{r,s}$ be the $(r,s)$ entry of the matrix $GG^H$ then it can readily be shown that
\begin{equation}
\begin{aligned}
a_{r,s}
&=\frac{1}{k}\sum_{m=0}^{\alpha-1}e^{jm(\theta_r-\theta_s)} +\frac{1}{k} \sum_{m=n-\beta}^{n-1}e^{jm(\theta_r-\theta_s)},\\
\end{aligned}
\label{eq:Toep}
\end{equation}
in which $\theta_x = \frac{2\pi}{n}(x-1)$. From this equation, it is clear that $a_{r,s}=a_{r+i,s+i}$;
that is, the elements of each diagonal are equal, which means that  $GG^H$ is a {\it Toeplitz} matrix.
In addition, we can check that $a_{r,n}=a_{r+1,1}$, i.e., the last entry in each row is equal to the first entry of the next row.
This proves that the Toeplitz matrix $GG^H$ is {\it circulant} as well \cite{gray2006toeplitz}.
Also, a quick look at \eqref{eq:Toep} reveals that the elements of the principal diagonal $(r=s)$ are equal to $1$.  Similarly, one
can see that for any  $1 < p < n$, the square matrix $G_{p\times k}G_{p\times k}^H$ is also a Toeplitz matrix; it is
not necessarily circulant, however.
\end{proof}

Removing $W_k$ from \eqref{eq:G1} we end up with a complex $G$, representing a complex BCH-DFT code. In such a code,
$\alpha$ and $\beta$ can be any nonnegative integers such that $\alpha + \beta =k$.

\begin{rem}
Lemma~\ref{lem1} also holds for complex BCH-DFT codes.
\end{rem}
\noindent As a result, all properties that we prove in the remainder of this paper are valid for ``any''
BCH-DFT code, although they are formally proved for ``real'' BCH-DFT codes,  which we simply refer to as ``DFT codes'' or ``DFT frames'' hereafter.

\subsection{Systematic DFT Frames}
\label{sec:sys}
In the context of channel coding, there is a special interest in
systematic codes so as to simplify the decoding algorithm. This is more pronounced in
``parity-based" DSC as it requires distinction between parity and data.
The parity-based approach becomes even more worthwhile in the DSC that uses DFT codes
because it is more ``efficient'' than the syndrome-based approach \cite{Vaezi2011DSC}.
This is due to the fact that syndrome, even in a real DFT code, is a complex vector
whereas parity is real. This encourages
a systematic generator matrix for DFT codes.

A systematic generator matrix for a real BCH-DFT code can be obtained by \cite{Vaezi2011DSC}
  \begin{align}
  G_{\mathrm{sys}}=GG_{k}^{-1},
\label{eq:Gsys2}
\end{align}
in which $G_{k}$ is a submatrix (subframe) of $G$ including $k$ arbitrary rows of $G$.
Note that $G_{k}$ is invertible since it can be represented as $G_{k} =\sqrt{\frac{n}{k}} W_{k \times n}^H \Sigma W_{k} = V_{k}^H W_{k}$,
in which $V_{k}^H \triangleq \sqrt{\frac{n}{k}} W_{k \times n}^H \Sigma$ is a Vandermonde matrix. Remember that $W_{k}$ is also invertible as it is a DFT matrix. This proves that
systematic DFT frames exist for any DFT frame. It also shows that there are many (but, a finite number of) systematic frames for each frame, because the rows of $G_{k}$ can be arbitrarily chosen
from those of $G$. These systematic frames differ in the ``position'' of systematic samples (input data) in resulting codewords. This implies that the parity samples are not restricted to form a consecutive block in codewords. Such a degree of freedom is useful in the sense that one can find the
most suitable systematic frames for specific applications (e.g.,
one with the smallest reconstruction error.)

\section{Main Results on the Extreme Eigenvalues}
\label{sec:eig}

From rate-distortion theory, it is well known that the rate required to transmit a source, with a given distortion, increases as
the variance of the source becomes larger \cite{Cover}. Particularly, for Gaussian sources this relation
is logarithmic with variance, under the mean-squared error (MSE) distortion measure.
In DSC that uses real-number codes \cite{Vaezi2011DSC}, since coding is performed before quantization,
the variance of transmitted sequence depends on the behavior of the
encoding matrix.
 In view of rate-distortion theory, it makes a lot of sense
to keep this variance as small as possible. Not surprisingly, we will show that using a tight
frame (tight $G_{\mathrm{sys}}$) for encoding is optimal.

Let $\bm{x}$ be the message vector and $\bm{y}=G_{\mathrm{sys}}\bm{x}$ represent the codevector generated using the systematic frame, then
the variance of  $\bm{y}$ is given as

  \begin{equation}
\begin{aligned}
\sigma^2_y &=  \frac{1}{n} \mathbb{E}\{\bm{y}^H\bm{y}\} =  \frac{1}{n} \mathbb{E}\{ \bm{x}^H G_{\mathrm{sys}}^H  G_{\mathrm{sys}}\bm{x}  \}  \\
&= \frac{1}{n} \sigma^2_x \tr{ (G_{\mathrm{sys}}^H  G_{\mathrm{sys}})},
\label{eq:vary}
\end{aligned}
\end{equation}
and
 \begin{equation}
\begin{aligned}
 \tr\left(G_{\mathrm{sys}}^HG_{\mathrm{sys}}\right) &= \tr\left(G_k^{-1H}G^HGG_k^{-1}\right)  \\
  &= \frac{n}{k}\tr\left((G_k G_k^H)^{-1}\right)  \\
  &= \frac{n}{k} \tr\left((V_{k}^HV_{k})^{-1}\right) \\
  &= \frac{n}{k}\sum_{i=1}^k\frac{1}{\lambda_i},
\label{eq:G7}
\end{aligned}
\end{equation}
in which $\lambda_1 \geq \lambda_2 \geq \cdots \geq \lambda_k>0$ are the eigenvalues of $G_{k}G_{k}^H$ (or $V_{k}^HV_{k}$ equivalently).

This shows that the variance of codevectors, generated by a systematic frame, depends on the submatrix $G_{k}$ which is used to create $G_{\mathrm{sys}}$. $G_{k}$, in turn,
is fully known once the position of systematic (data) samples is fixed in the codevector. In other words, the ``position'' of systematic samples, determines
the variance of codevectors generated by a systematic DFT frame.
From \eqref{eq:vary}, \eqref{eq:G7}, to minimize the effective range of transmitted signal, we need to do the following optimization problem.
\begin{equation}
\begin{aligned}
& \underset{\lambda_i}{\text{minimize}}
& & \sum_{i=1}^k\frac{1}{\lambda_i}
& \text{s.t.}
& & \sum_{i=1}^k\lambda_i=k, \,\, \lambda_i>0,
\end{aligned}
\label{eq:Omin}
\end{equation}
where, the constraint $\sum_{i=1}^k\lambda_i=k$ is achieved in light of Lemma~\ref{lem1} and \eqref{eq:Schur-Horn2}.

By using the Lagrangian method, we can show that the optimal eigenvalues are $\lambda_i=1$; this implies a tight frame \cite{goyal2001quantized}.
In the sequel, we analyze
the eigenvalues of $G_{p\times k}G_{p\times k}^H$, $p \leq n$, that helps us characterize tight systematic frames, so as to
minimize the variance of transmitted codevectors.

\begin{thm}
Let $G_{p\times k}, 1 \leq p \leq n$ be any $p\times k$ submatrix of $G$. Then, the smallest eigenvalue of
$G_{p\times k}G_{p\times k}^H$ is no more than one, and the largest eigenvalue of
$G_{p\times k}G_{p\times k}^H$ is at least one.
\label{thm1}
\end{thm}

\begin{proof}
From Lemma~\ref{lem1}, we know that all principal diagonal entries of $G_{p\times k}G_{p\times k}^H$
 are unity. As a result, using the Schur-Horn inequality in \eqref{eq:Schur-Horn1},
 we obtain $\lambda_{\min}(G_{p\times k}G_{p\times k}^H) \leq 1 \leq  \lambda_{\max}(G_{p\times k}G_{p\times k}^H)$.
 This proves the claim. Also, note that for any $G_{p\times k}$,
 $\lambda_{\max}(G_{p\times k}G_{p\times k}^H)=\lambda_{\max}(G_{p\times k}^HG_{p\times k})$.
\end{proof}

\noindent We use the above results to find better bounds for the extreme eigenvalues of $G_{k}G_{k}^H$ in the following theorem.

\begin{thm}
For any $G_{k}$, a square submatrix of $G$ in \eqref{eq:G1} in which $n\neq Mk$, the smallest (largest) eigenvalue of
$G_{k}G_{k}^H$ is strictly upper (lower) bounded by $1$.
\label{thm2}
\end{thm}

\begin{proof}
We first investigate a bound for the smallest eigenvalue. Let $n=Mk+l, 0<l<k$, then $G$ can be partitioned  as $G=[G_{k}^{H} \,|\, G_{k}^{1H} \,|\,\cdots \,
|\,G_{k}^{(M-1)H}\,|\,G_{k\times l}^{MH}]^H$.
In general, $G_{k},G^{1}_{k}, \hdots ,G^{M-1}_{k}$ and $G_{k\times l}^{M}$ include arbitrary rows of $G$, hence they have
different spectrums, i.e., different sets of eigenvalues.
%
We consider the case with largest $\lambda_{k}$ for $G_{k}^HG_{k}$; this occurs when $G_{k}$ consist of the rows of $G$
such that the distance between each two successive rows is as large as possible and at least $M$.
The latter guarantees existence of $G_{k}^{1}, \hdots ,G_{k}^{M-1}$
such that $G_{k}^{mH}G_{k}^{m}$, for any $1 \leq m\leq M-1$, has the same spectrum as $G_{k}^HG_{k}$.
To find the row indices corresponding to $G^{m}_{k}$,
we can simply add $m$ to each row index of $G_{k}$.
Then, to show these matrices have the same spectrum, we use Lemma~3 \cite{tucci2011eigenvalue} which states that
 any Hermitian $n\times n$ matrices $E$ and $F$ with
 $F_{i,j} = \frac{c_i}{c_j}E_{i,j}$ have the same eigenvalues.
Now, given a $G_{k}$, one can verify that $(G^{1}_{k})_{i,j} =  a_j (G_{k})_{i,j}$ and thus
 $(G_{k}^{1})_{i,j}^H =  a_i^\ast (G_{k})_{i,j}^H =  1/a_i (G_{k})_{i,j}^H$. Therefore,
 $G_{k}^{1H} G^{1}_{k}$ and $G_{k}^{H} G_{k}$
 have the same spectrum. The same argument is valid for $G^{2}_{k}, \hdots ,G^{M-1}_{k}$.
 Next, we see that $G^HG= A+B$ in which $A= G^{H}_{k} G_{k}+ \cdots +G^{(M-1)H}_{k} G^{M-1}_{k}$ and
$B=G^{MH}_{k \times l}G^M_{k \times l}$. Then, in consideration of the above discussion, $\lambda_{i}(A)= M\lambda_{i}(G_{k}^H G_{k})$ for any $1\leq i \leq k$.
Hence, from \eqref{eq:Weyl2}, for $i=1, j=k$, we will have
  \begin{equation}
\begin{aligned}
  \lambda_{k}(A)+\lambda_{1}(B)&\leq   \lambda_{1}(A+B)\\
  \Leftrightarrow  M\lambda_{k}(G_{k}^H G_{k}) &\leq\frac{n}{k} -\lambda_{1}(B) \\
    \Leftrightarrow  \lambda_{k}(G_{k}^H G_{k}) &\leq    \frac{\frac{n}{k}-1}{M} = \frac{\frac{n}{k}-1}{\lfloor\frac{n}{k} \rfloor }<1,
\end{aligned}
\label{eq:thm2}
 \end{equation}
where the last line follows using $\lambda_{1}(B)\geq 1$ from Theorem~\ref{thm1}.
This completes the proof for the worst case, i.e., the largest possible $\lambda_{k}(G_{k}^H G_{k})$,
and implies that \eqref {eq:thm2} holds for any other $G_{k}$.
Hence, the fist part of the proof is completed; that is, the smallest eigenvalue of
$G_{k}^HG_{k}$ where $G_{k}$ is an arbitrary square submatrix of
$G$ in \eqref{eq:G1} with $n\neq Mk$
is strictly less than one.

Finally, knowing that $\sum_{\substack i=1}^{k}{\lambda_{i}(G_{k}^H G_{k})} =\sum_{\substack i=1}^{k}a_{ii}=k$ and
using \eqref{eq:thm2},
it is obvious that $\lambda_{1}(G_{k}^H G_{k})> 1$. This proves the bound for the largest eigenvalue.

\end{proof}

Theorem~\ref{thm2} implies that for $n\neq Mk$ we cannot have ``tight'' systematic frames.
This is due to the fact that for a tight frame with frame operator $F$,
$\lambda_{\min}(F^HF)=\lambda_{\max}(F^HF)$; i.e., the eigenvalues of $F^HF$ are equal \cite{goyal2001quantized}.

\begin{cor}
\label{cor2}
For $n \neq Mk$, where $M$ is a positive integer, ``tight'' systematic DFT frames do not exist.
\end{cor}

Note that systematic DFT frames are not necessarily tight for $n = Mk$. Evaluating the performance of systematic frames in the next section,
we prove that tight systematic DFT frames exist for $n = Mk$ and show how to construct such frames.

\section{Performance Analysis and Classification of Systematic Frames }
\label{sec:class}
\subsection{Performance Evaluation}

In this section, we analyze the performance
of quantized systematic DFT codes using the quantization model proposed in \cite{goyal2001quantized},
which assumes that noise components are uncorrelated and each noise component $q_i$
has mean $0$ and variance $\sigma_q^2$.
We assume the quantizer range covers the
dynamic range of all codevectors encoded using the systematic DFT code in \eqref{eq:Gsys2}.

The codevectors are generated by $\bm{y}=G_{\mathrm{sys}}\bm{x}$.
We also consider linear reconstruction of \mbox{\boldmath$x$} form
$\bm{y}$ using the pseudoinverse \cite{goyal2001quantized} of $G_{\mathrm{sys}}$, which is defined
as $G_{\mathrm{sys}}^\dagger = (G_{\mathrm{sys}}^HG_{\mathrm{sys}})^{-1}G_{\mathrm{sys}}^H$.
It is easy to check that $G_{\mathrm{sys}}^\dagger= \frac{k}{n}G_{k}G^H$, then the linear reconstruction is given by
\begin{equation}
\begin{aligned}
\bm{x} = G_{\mathrm{sys}}^\dagger \bm{y}
= \frac{k}{n}G_{k}G^H\bm{y}.
\label{eq:G3}
\end{aligned}
\end{equation}
Now, suppose we want to estimate \mbox{\boldmath$x$}
from $\hat{\bm{y}}  = G_{\mathrm{sys}}\bm{x}+ \bm{q}$, where $\bm{q}$ represents quantization error. From \eqref{eq:G3} we obtain
\begin{align}
\hat{\bm{x}} = \frac{k}{n}G_{k}G^H \hat{\bm{y}}=\bm{x}+\frac{k}{n}G_{k}G^H\bm{q}.
\label{eq:G4}
\end{align}
Then, the mean-squared reconstruction error, due to the quantization noise, using a systematic frame can be written as
\begin{equation}
\begin{aligned}
\mathop{\mathrm{MSE_{q}}} &= \frac{1}{k} \mathbb{E}\{\|\hat{\bm{x}} -\bm{ x}\|^2\}  = \frac{1}{k}  \mathbb{E}\{\|G_{\mathrm{sys}}^{\dagger}\bm{q}\|^2\}  \\
&= \frac{1}{k} \mathbb{E}\{\bm{q}^H G_{\mathrm{sys}}^{\dagger H}G_{\mathrm{sys}}^{\dagger}\bm{q}\}  \\
&= \frac{1}{k}  \sigma_q^2 \tr\left(G_{\mathrm{sys}}^{\dagger H}G_{\mathrm{sys}}^{\dagger}\right)\\
&= \frac{k}{n^2} \sigma_q^2 \tr\left(GG_{k}^H G_{k}G^H\right)\\
&= \frac{k}{n^2} \sigma_q^2 \tr\left(G_{k}^H G_{k}G^HG\right)\\
&= \frac{1}{n}  \sigma_q^2 \tr\left(G_{k}^H G_{k}\right)= \frac{k}{n} \sigma_q^2 ,
\label{eq:G6}
\end{aligned}
\end{equation}
where the last step follows because of Lemma~\ref{lem1}. The above analysis indicates that the
MSE is the same for all systematic DFT frames of same size, provided that the effective range codevectors  generated by different $G_{\mathrm{sys}}$ is equal.
This implies a same $\sigma_q^2$ for a given number of quantization levels.
However, for a fixed number of quantization levels, $\sigma_q^2$ depends on the
variance of transmitted codevectors, which, in turn, varies for different systematic frames, as shown in \eqref{eq:vary}.

As we discussed in Section~\ref{sec:eig}, the optimal $G_{\mathrm{sys}}$ is achieved
from the optimization problem \eqref{eq:Omin}.
Similarly, to find the worst $G_{\mathrm{sys}}$, we can {\it maximize} \eqref{eq:Omin} instead of minimizing it.
The optimal eigenvalues are known to be $\lambda_i=1$. But, how can we find the corresponding $G_{\mathrm{sys}}$, or $G_{k}$ equivalently?

%

We approach this problem by studying another optimization problem. By using the Lagrangian method, one can check the optimal arguments of the optimization problem in \eqref{eq:Omin} are equal to those of
\begin{equation}
\begin{aligned}
& \underset{\lambda_i}{\text{maximize}}
& & \prod_{i=1}^k \lambda_i
& \text{s.t.}
& & \sum_{i=1}^k\lambda_i=k, \,\, \lambda_i>0,
\end{aligned}
\label{eq:Omax}
\end{equation}
in which $\{\lambda_i\}_{i=1}^k$ are the eigenvalues of $G_{k}G_{k}^H$ (or $V_{k}^HV_{k}$).
In other words, subject to the above constraints
 \begin{equation}
\begin{aligned}
& \underset{\lambda_i}{\operatorname{argmin}}  \sum_{i=1}^k\frac{1}{\lambda_i} = \underset{\lambda_i}{\operatorname{argmax}} \prod_{i=1}^k \lambda_i. \\
\end{aligned}
\label{eq:Oeq}
\end{equation}
Problem \eqref{eq:Omax} has the maximum of 1 and infimum of 0. Then, considering that $\prod_{i=1}^k \lambda_i = \det (V_{k}^HV_{k})=  \det (G_{k}G_{k}^H)$,
we conclude that the ``best'' submatrix is the one with the largest determinant (possibly 1) and the ``worst'' submatrix is the one with smallest determinant.
Next, we evaluate the determinant of $V_{k}^HV_{k}$ so as to find the matrices corresponding to these extreme cases.

\subsection{The Best and Worst Systematic Frames}
\label{sec:class1}
In this section, we first evaluate the determinate of $WW^H$ where $W$ is the Vandermonde matrix with unit complex entries as defined in
\eqref{eq:Vand2}. From \eqref{eq:detV0} we can write
 \begin{equation}
\begin{aligned}
 \det(WW^H) &= \frac{1}{n^n}\prod_{\substack 1\leq p<q\leq n}{|e^{i\theta_p}-e^{i\theta_q}|^2} \\
  &= \frac{1}{n^n}\prod_{\substack 1\leq p<q\leq n}{4 \sin^2 \frac{\pi}{n}(q - p)} \\
  &= \frac{2^{n(n-1)}}{n^n}\prod_{\substack r=1}^{n-1}{\left( \sin^2 \frac{\pi}{n}r\right)^{n-r}}
 \label{eq:detV}
\end{aligned}
\end{equation}
in which $\theta_x = \frac{2\pi}{n}(x-1),  r=q - p$, and $n(n-1)/2$ is the total number of terms that
satisfy $1\leq p<q\leq n$.
But, we see that $W$ is a DFT matrix, and thus,
its determinant must be 1. Therefore, we have
 \begin{align}
 \prod_{\substack r=1}^{n-1}{\left( \sin^2 \frac{\pi}{n}r\right)^{n-r}}=\frac{n^n}{2^{n(n-1)}}.
 \label{eq:formula}
 \end{align}

%

 The above analysis helps us evaluate the determinant of $V_{k}$  or  $G_{k}$, defined in \eqref{eq:Gsys2}.
Let $\mathcal I_{r}=\{i_{r_1}, i_{r_2}, \hdots, i_{r_k}\}$ be those rows of $G$ used to
build $G_{k}$.  Also, without loss of generality, assume $i_{r_1}< i_{r_2}< \cdots < i_{r_k}$. Clearly, $i_{r_1}\geq 1, i_{r_k}\leq n$, and
we obtain
\begin{equation}
\begin{aligned}
 \det(V_kV_k^H) &= \frac{1}{k^k}\prod_{\substack {1\leq p<q\leq n \\ p, q \in \mathcal I_{r}}}{|e^{i\theta_p}-e^{i\theta_q}|^2} \\
  &= \frac{1}{k^k}\prod_{\substack {1\leq p<q\leq n \\ p, q \in \mathcal I_{r}}}{4 \sin^2 \frac{\pi}{n}(q - p)}.
 \label{eq:detVk}
 \end{aligned}
\end{equation}
Then, since $\sin \frac{\pi}{n}u = \sin \frac{\pi}{n}(n-u) $, one can see that this determinant depends on the
circular distance between rows in $\mathcal I_r$. For a matrix with $n$ rows, we define the circular distance between rows $p$ and $q$
as $\min{\{|q - p|, n-|q - p|\}}$. In this sense, for example, the distance
between rows $1$ and  $n$ is one, i.e., they are circularly successive.
Now, we can see that \eqref{eq:detVk} is minimized when the selected rows are (circularly) successive.
Note that $\sin u$ is strictly increasing for  $u \in [0, \pi/2]$ and the circular distance cannot be greater than $n/2$, in this problem.

In such circumstances where all rows in ${\mathcal I_r}$ are (circularly) successive, \eqref{eq:detVk} is minimal and reduces to
 \begin{align}
 \det(V_kV_k^H) = \frac{2^{k(k-1)}}{k^k}\prod_{\substack r=1}^{k-1}{\left( \sin^2 \frac{\pi}{n}r\right)^{k-r}}.
 \label{eq:detVk2}
 \end{align}
The other extreme case comes up when $n=Mk$ ($M$ is a positive integer) provided that $G_{k}$ consists of every $M$th row of $G$.
In such a case \eqref{eq:detVk} simplifies to $1$ because

\begin{equation}
\begin{aligned}
 \det(V_kV_k^H) &=  \frac{2^{k(k-1)}}{k^k}\prod_{\substack r=1}^{k-1}{\left( \sin^2 \frac{\pi}{n}Mr\right)^{k-r}}  \\
  &=  \frac{2^{k(k-1)}}{k^k}\prod_{\substack r=1}^{k-1}{\left( \sin^2 \frac{\pi}{k}r\right)^{k-r}} =1,
 \label{eq:detVm}
\end{aligned}
\end{equation}
where the last step follows from \eqref{eq:formula}. Recall that this gives the best $V_{k}$ (and equivalently $G_{k}$), in light of \eqref{eq:Omax}.
For such a $G_{k}$, it is easy to see that $G_{\mathrm{sys}}$ stands for a ``tight'' systematic frame and minimizes the MSE for a given number of quantization levels.
Effectively, such a frame is performing {\it integer oversampling}. There are $M$ such frames; they all have the same spectrum, though.

\subsection{Numerical Examples}
\label{sec:num}

Numerical calculations confirm that ``evenly" spaced data samples gives rise to
systematic frames with the best performance.
When a systematic code is doing integer oversampling, 
we end up with tight systematic frames. The first code in Table~\ref{table1}
is an example of this case. When $n\neq Mk$, data samples cannot be equally
spaced; however, as it can be seen from the second code in Table~\ref{table1},
still the best performance is achieved when they are as equally spaced as possible.
Note that, circular shift of codewords pattern does not change the spectrum of corresponding matrices.
For example, in the $(7, 5)$ code, codewords with pattern $\times -\times \times\times -\times $ and $\times\times -\times\times -\times $
have the same properties. Also, reversal of a codeword yields a codeword with similar properties
(e.g., $\times\times -\times - -$ is shifted version of reversed $\times\times --\times   -$).

\begin{table}[!t]
\caption{Eigenvalues structure for two systematic DFT frames with different codeword patterns.
A ``$\times$'' and ``$-$'' respectively represent data (systematic) and parity samples.} \label{table1}
\centering
\scalebox{.87}{
\begin{tabular}{llcccccc}
\toprule %
\multicolumn {2}{c}{$\mathop{\mathrm{Code \quad\quad Codeword }} \qquad\;  $} & $\lambda_{\min}$  & $\lambda_{\max}$ & $\sum_{i=1}^k 1/\lambda_i$ & $\prod_{i=1}^k\lambda_i$ \\
\multicolumn {2}{c}{$\mathop{\mathrm{\qquad\quad\quad\quad  patern}} \qquad\;  $} &   &  &  &  \\ \toprule \toprule %
& $\times\times\times  - - -$  & $0.0572$ & $1.9428$ & $19$ & $0.1111$ &   \\
\cmidrule (r){3-7}
\multirow {2}*{$(6, 3)$}
& $\times\times -\times  - -$ & $0.2546$ & $1.7454$ & $5.5$ & $0.4444$ &   \\
\cmidrule (r){3-7}
& $\times\times --\times   -$ & $0.2546$ & $1.7454$ & $5.5$ & $0.4444$ &   \\
\cmidrule (r){3-7}
& $\times -\times -\times-$ & $1$ & $1$ & $3$ & $1$ &    \\
\midrule
& $\times\times\times \times\times - -$  & $0.0396$ & $1.4$ & $28.70$ & $0.0827$ &   \\
\cmidrule (r){3-7}
\multirow {2}*{$(7, 5)$}
& $\times\times \times\times -\times - $ & $0.1506$ & $1.4$ & $10.32$ & $0.2684$ &   \\
\cmidrule (r){3-7}
& $\times\times -\times\times -\times $ & $0.3110$ & $1.4$ & $7.40$ & $0.4173$ &   \\
\cmidrule (r){3-7}
& $\times -\times \times\times -\times $ & $0.3110$ & $1.4$ & $7.40$ & $0.4173$ &   \\
\bottomrule
\end{tabular}
}

\end{table}

\section{Conclusions}
\label{sec:con}
Systematic DFT frames as well as the approach to make such a frame
out of the generator matrix of a BCH-DFT code has been introduced.
Further, we found the conditions for which a systematic DFT frame 	
can be tight, too.
We then related the performance of these frames to
the position of systematic samples in the codevector. The analysis shows that evenly spaced systematic (parity)
samples result in the minimum reconstruction error, whereas the worst performance is achieved when systematic samples are circularly successive.
Finally, we found the conditions
for which a DFT frame becomes both systematic and tight.

%
%

\end{document}